\documentclass[a4paper,10pt,twocolumn]{article}
\usepackage[footnotesize]{caption}
\usepackage[caption=false,font=footnotesize]{subfig}
\usepackage{mathptmx,graphicx,amsmath,amsfonts,amsthm,amssymb,url,cite}
\usepackage[margin={0.575in,0.8in}]{geometry}

\newcommand{\expn}{E}
\newcommand{\pr}{P}
\newcommand{\eps}{\epsilon}

\newtheorem{theorem}{Theorem}
\newtheorem{definition}{Definition}
\newtheorem{lemma}{Lemma}
\newtheorem{corollary}{Corollary}

\begin{document}

\title{Dynamic Tardos Traitor Tracing Schemes}
\author{Thijs Laarhoven\footnote{T. Laarhoven, B. \v{S}kori\'{c}, and B. de Weger are with the Department of Mathematics and Computer Science, Eindhoven University of Technology, P.O. Box 513, 5600 MB Eindhoven, The Netherlands. \protect\\
E-mail: \{t.m.m.laarhoven,b.skoric,b.m.m.d.weger\}@tue.nl.} \and Jeroen Doumen\footnote{J. Doumen and P. Roelse are with Irdeto, P.O. Box 3047, 2130 KA Hoofddorp, The Netherlands.\protect\\
E-mail: \{jdoumen,peter.roelse\}@irdeto.com.} \and Peter Roelse\footnotemark[2] \and Boris \v{S}kori\'{c}\footnotemark[1] \and Benne de Weger\footnotemark[1]}
\date{\today}

\maketitle

\begin{abstract}
We construct binary dynamic traitor tracing schemes, where the number of watermark bits needed to trace and disconnect any coalition of pirates is quadratic in the number of pirates, and logarithmic in the total number of users and the error probability. Our results improve upon results of Tassa, and our schemes have several other advantages, such as being able to generate all codewords in advance, a simple accusation method, and flexibility when the feedback from the pirate network is delayed.
\end{abstract}

\section{Introduction}
\label{sec:Introduction}

To protect digital content from unauthorized redistribution, distributors embed watermarks in the content such that, if a customer distributes his copy of the content, the distributor can see this copy, extract the watermark and see which user it belongs to. By embedding a unique watermark for each different user, the distributor can always determine from the detected watermark which of the customers is guilty. However, several users could cooperate to form a coalition, and compare their differently watermarked copies to look for the watermark. Assuming that the original data is the same for all users, the differences they detect are differences in their watermarks. The colluders can then distort this watermark, and distribute a copy which matches all their copies on the positions where they detected no differences, and has some possibly non-deterministic output on the detected watermark positions. Since the watermark does not match any user's watermark exactly, finding the guilty users is non-trivial. 

In this paper we focus on the problem of constructing efficient collusion-resistant schemes for tracing pirates, which involves finding a way to choose watermark symbols for each user (the traitor tracing code) and a way to trace a detected copy back to the guilty users (an accusation algorithm). In particular, we will focus on the application of such schemes in the dynamic setting, where the pirate output is detected in real-time, before the next watermark symbols are embedded in consecutive segments of the content. We will show that by building upon the (static) Tardos scheme \cite{tardos03}, we can construct efficient and flexible dynamic traitor tracing schemes. The number of watermark symbols needed in our schemes is a significant improvement compared to the scheme of Tassa \cite{tassa05}, and our schemes can be easily adjusted when the model is slightly different from the standard dynamic traitor tracing model \cite{berkman01,fiat01,roelse11,tassa05}.

\subsection{Model}
\label{sub:Introduction-Model}

Let us first formally describe the mathematical model for the problem discussed in this paper. First, some entity called the distributor controls the database of watermarks and distributes the content. The recipients, each receiving a watermarked copy of the content, are referred to as users. We write $U = \{1, \ldots, n\}$ for the set of all users, and we commonly use the symbol $j$ for indexing these users. For the watermarks, we refer to the sequence of watermarking symbols assigned to a user $j$ by the vector $\vec{X}_j$, which is also called a codeword. We write $\ell$ for the total number of watermark symbols in a codeword, so that each codeword $\vec{X}_j$ has length $\ell$, and we commonly use the symbol $i$ to index the watermark positions. We write $\mathcal{X}$ for the algorithm used to generate the codewords $\vec{X}_j$. In this paper we only focus on watermark symbols from a binary alphabet, so that $(\vec{X}_j)_i \in \{0,1\}$ for all $i,j$. A common way to represent the traitor tracing code is by putting all codewords $\vec{X}_j$ as rows in a matrix $X$, so that $X_{j,i} = (\vec{X}_j)_i$ is the symbol on position $i$ of user $j$.

After assigning a codeword to each user, the codewords are embedded in the data as watermarks. The watermarked copies are sent to the users, and some of the users (called the pirates or colluders) collude to create a pirate copy. The pirates form a subset $C \subseteq U$, and we use $c = |C|$ for the number of pirates in the coalition. The pirate copy has some distorted watermark, denoted by $\vec{y}$. We assume that if on some position $i$ all pirates see the same symbol, they output this symbol. This assumption is known in the literature as the marking assumption. On other positions we assume pirates simply choose one of the two symbols to output. This choice of pirate symbols can be formalized by denoting a pirate strategy by a (probabilistic) function $\rho$, which maps a code matrix $X$ (or the part of the matrix visible to them) to a forgery $\vec{y}$. After the coalition generates a pirate copy, we assume the distributor detects it and uses some accusation algorithm $\sigma$ to map the forgery $\vec{y}$ to some subset $\sigma(\vec{y}) = \hat{C} \subseteq U$ of accused users. These users are then disconnected from the system. Ideally $\hat{C} = C$, but this may not always be achievable.

\paragraph{Static schemes.} We distinguish between two types of schemes. In static schemes, the process ends after one run of the above algorithm with a fixed codelength $\ell$, and the set $\hat{C}$ is the final set of accused users. So the complete codewords are generated and distributed, the pirates generate and distribute a pirate copy, and the distributor detects this output and calculates the set of accused users. In this case an elementary result is that one can never have any certainty of catching all pirates. After all, the coalition could decide to sacrifice one of its members, so that $\vec{y} = \vec{X}_j$ for some $j \in C$. Then it is impossible to distinguish between other pirates $j' \in C \setminus \{j\}$ and innocent users $j' \in U \setminus C$. However, static schemes do exist that achieve catching at least one guilty user and not accusing any innocent users with high probability. The original Tardos scheme \cite{tardos03} belongs to this class of schemes.

\paragraph{Dynamic schemes.} The other type of scheme is the class of dynamic schemes, where the process of sending out symbols, detecting pirate output and running an accusation algorithm is repeated multiple times. In this case, if a user is caught, he is immediately cut off from the system and can no longer access the content. These dynamic scenarios for example apply to live broadcasts, such as pay-tv. The distributor broadcasts the content, while the pirates directly output a pirate copy of the content. The distributor then listens in on this pirate broadcast, extracts the watermarks, and uses this information for the choice of watermarks for the next segment of the content. We assume that the pirates always try to keep their broadcast running, so that if one of the pirates is disconnected, the other pirates will take over. Ideally one demands that the set of accused users always matches the exact coalition, i.e.\ $\hat{C} = C$, and with dynamic schemes we can also achieve this with high probability, as we will see later. The new schemes we present in this paper belong to this class of schemes.

As mentioned earlier, we call static schemes successful if with high probability, at least one guilty user is caught, and no innocent users are accused. With dynamic schemes one can catch all pirates, so we only call such schemes successful if with high probability, all pirates are caught and no innocent users are accused. This leads to the following definitions of soundness and static/dynamic completeness.

\begin{definition}[Soundness and completeness] \label{def:Secureness}
Let $(\mathcal{X}, \sigma)$ be a traitor tracing scheme, let $c_0 \geq 2$ and let $\eps_1,\eps_2 \in (0,1)$. Then this scheme is called $\eps_1$-sound, if for all coalitions $C \subseteq U$ and pirate strategies $\rho$, the probability of accusing one or more innocent users is bounded from above by
  \begin{align*}
  \pr(\hat{C} \not\subseteq C) \leq \eps_1.
  \end{align*}
A static traitor tracing scheme $(\mathcal{X}, \sigma)$ is called static $(\eps_2, c_0)$-complete, if for all coalitions $C \subseteq U$ of size at most $c_0$ and for all pirate strategies $\rho$, the probability of not catching \textit{any} pirates is bounded from above by
  \begin{align*}
  \pr(C \cap \hat{C} = \emptyset) \leq \eps_2.
  \end{align*}
Finally, a dynamic traitor tracing scheme $(\mathcal{X}, \sigma)$ is called dynamic $(\eps_2, c_0)$-complete, if for all coalitions $C \subseteq U$ of size at most $c_0$ and for all pirate strategies $\rho$, the probability of not catching \textit{all} pirates is bounded from above by
  \begin{align*}
  \pr(C \not\subseteq \hat{C}) \leq \eps_2.
  \end{align*}
\end{definition}

Note that we distinguish between $c$, the \textit{actual} collusion size, and $c_0$, the \textit{estimated} collusion size used by the distributor to build the traitor tracing scheme. Since $c$ is usually unknown, the distributor has to make a guess $c_0 \approx c$, which has to be sufficiently large to guarantee security, and sufficiently small to guarantee efficiency. 

In the following sections we will omit the $c_0$ in the completeness property if the parameter is implicit. Similarly, when $\eps_1$ or $\eps_2$ is implicit, we simply call a scheme sound or complete. As we will see later, in the schemes discussed in this paper, $\eps_1/n$ and $\eps_2$ are closely related. We will use the notation $\eta = \ln(\eps_2)/\ln(\eps_1/n)$ to denote the log ratio of these error probabilities. In most practical scenarios we have $\eps_1/n < \eps_2$, so usually $\eta \in (0,1)$.

\subsection{Related work}
\label{sub:Introduction-RelatedWork}

The schemes in this paper all build upon the Tardos scheme \cite{tardos03}, introduced in 2003. This is an efficient static traitor tracing scheme, and it was the first scheme to achieve $\eps_1$-soundness and $(\eps_2, c_0)$-completeness with a codelength of $\ell = O(c_0^2 \ln(n/\eps_1))$. In the same paper it was proved that this order codelength is asymptotically optimal for large $c$. The original Tardos scheme had a codelength of $\ell = 100 c_0^2 \ln(n/\eps_1)$, and several improvements of the Tardos scheme have been suggested to reduce the constant before the $c_0^2 \ln(n/\eps_1)$. We mention two in particular: the improved analysis done by Blayer and Tassa \cite{blayer08}; and the introduction of a symmetric score function by \v{S}kori\'{c} et al. \cite{skoric08}. Laarhoven and De Weger combined these improvements \cite{laarhoven11} to get even shorter codelength constants. For $c_0 \geq 2$ and $\eta \leq 1$, this construction gives codelengths of $\ell < 24 c_0^2 \ln(n/\eps_1)$, with the constant further decreasing as $c_0$ increases or $\eta$ decreases. For asymptotically large $c_0$, this construction leads to codelengths satisfying $\ell = [\frac{\pi^2}{2} + O(c_0^{-1/3})] c_0^2 \ln(n/\eps_1)$. The symmetric Tardos scheme and its properties are discussed in Section~\ref{sec:Preliminaries}.

For the dynamic setting, we mention four papers. In 2001, Fiat and Tassa~\cite{fiat01} constructed a deterministic scheme, i.e., a scheme with $\eps_1 = \eps_2 = 0$. The number of symbols needed to catch pirates in that scheme is only $\ell = O(c \log n)$, but the alphabet size required is $q = 2c + 1$. In the same year, Berkman et al.~\cite{berkman01} proposed several deterministic schemes using a smaller alphabet of size $q = c + 1$, with codelengths ranging from $O(c^3\log_2(n))$ to $O(c^2 + c\log_2(n))$. In 2005, Tassa~\cite{tassa05} combined the dynamic scheme of Fiat and Tassa~\cite{fiat01} with the static scheme of Boneh and Shaw~\cite{boneh98}, to get a dynamic scheme using a binary alphabet, with a codelength of $\ell = O(c^4 \log_2(n) \ln(c/\eps_1))$. In the same paper it was suggested that using the Tardos scheme instead of the scheme of Boneh and Shaw as a building block may decrease the codelength by a factor $c$, thus possibly giving a codelength of $\ell = O(c^3 \log_2(n) \ln(c/\eps_1))$. In 2011, Roelse~\cite{roelse11} presented another deterministic scheme. As in the generalization of the scheme of Fiat and Tassa presented by Berkman et al.~\cite{berkman01}, in the scheme of Roelse the alphabet size equals $kc + 1$ with $k \geq 2$ and for a fixed value of $k$, the codelength is $O(c \log n)$. Moreover, the real-time computational cost and the bandwidth usage are logarithmic in $n$, instead of linear in $n$ as in the scheme of Fiat and Tassa and its generalization of Berkman et al.~\cite{berkman01}.

\subsection{Contributions and outline}
\label{sub:Introduction-Contributions}

First we show that the static Tardos scheme can be extended to a dynamic traitor tracing scheme in an efficient way, allowing us to catch the whole coalition instead of at least one colluder. This dynamic scheme has a codelength of $\ell = O(c c_0 \ln(n/\eps_1))$, where the constants only slightly increase compared to the constants of Laarhoven and De Weger~\cite{laarhoven11}. The adjustments do not influence the method of generating codewords, so these can still be generated in advance. 

To avoid the loss of efficiency caused by having to choose a value $c_0$, we then show how to create a $c_0$-independent ``universal" dynamic scheme that does not require a sharp estimate of $c$ as input. The property that the codewords can be generated in advance is left unchanged, while the scheme also has several advantages with respect to flexibility, detailed in Section~\ref{sec:Discussion}. The codelength of this scheme is also $\ell = O(c^2 \ln(n/\eps_1))$, thus improving upon the results of Tassa~\cite{tassa05} by roughly a factor $O(c^2)$ and upon the suggested improvement of Tassa by a factor $O(c)$.

The paper is organized as follows. First, we recall the construction of the static symmetric Tardos scheme and its properties in Section~\ref{sec:Preliminaries}. This scheme and its results will be used as the foundation for the dynamic Tardos scheme, which we present in Section~\ref{sec:DynamicTardos}. Then, in Section~\ref{sec:SemiDynamicTardos} we present a modification of the dynamic Tardos scheme when the setting is not fully dynamic. In Section~\ref{sec:UniversalTardos} we then present the universal Tardos scheme, which is an extension of the dynamic Tardos scheme that does not require a sharp bound on $c$ as input. In Section~\ref{sec:Discussion}, we discuss the results and argue that our schemes have several advantages with respect to flexibility as well. Finally, in Section~\ref{sec:OpenProblems} we list some open problems raised by our work.

This paper is mainly based on results from the first author's Master's thesis \cite{msc11}.

\section{Preliminaries: The Tardos scheme}
\label{sec:Preliminaries}

The results in the next sections all build upon results from the (static) symmetric Tardos scheme, so we first discuss this scheme here. Since the codeword generation of the schemes discussed in this paper all use (a variant of) the arcsine distribution, we also explicitly mention this distribution below. 

\subsection{Arcsine distribution}
\label{sub:Preliminaries-Arcsine}

The standard arcsine distribution function $F(p)$ on $[0,1]$, and its associated probability density function $f(p)$, are given by:
\begin{align}
  F(p) = \frac{2}{\pi} \arcsin(\sqrt{p}), \quad f(p) = \frac{1}{\pi \sqrt{p (1 - p)}}. \label{dist2}
\end{align}
This distribution function will be used in Section~\ref{sec:UniversalTardos}. In Sections~\ref{sec:Preliminaries}, \ref{sec:DynamicTardos} and \ref{sec:SemiDynamicTardos} we will use a variant of this distribution function, where the values of $p$ cannot be arbitrarily close to $0$ and $1$, as this generally leads to a high probability of accusing innocent users. Tardos~\cite{tardos03} therefore used the arcsine distribution with a certain small cutoff parameter $\delta > 0$, such that $p$ is always between $\delta$ and $1 - \delta$. By scaling $F$ and $f$ appropriately on this interval, this leads to the following distribution functions $F_{\delta}$ and associated probability density functions $f_{\delta}$:
\begin{align}
  F_{\delta}(p) &= \frac{2 \arcsin(\sqrt{p}) - 2 \arcsin(\sqrt{\delta})}{\pi - 4 \arcsin(\sqrt{\delta})}, \label{dist1} \\
	f_{\delta}(p) &= \frac{1}{(\pi - 4 \arcsin(\sqrt{\delta}))\sqrt{p (1 - p)}}. 
\end{align}
Note that taking $\delta = 0$ (i.e., using no cutoff) leads to $F_0(p) \equiv F(p)$.

\subsection{Construction}
\label{sub:Preliminaries-Construction}

The Tardos scheme, with parameters $d_{\ell}, d_z, d_{\delta}$ as used by Blayer and Tassa \cite{blayer08} and Laarhoven and De Weger \cite{laarhoven11}, and with the symmetric score function introduced by \v{S}kori\'{c} et al.~\cite{skoric08}, is described below. 


\begin{enumerate}
  \item \textbf{Initialization phase}
  \begin{enumerate}
    \item Take the codelength as $\ell = d_{\ell} c_0^2 \ln(n/\eps_1)$.
		\item Take the threshold as $Z = d_z c_0 \ln(n/\eps_1)$.
		\item Take the cutoff parameter as $\delta = 1/(d_{\delta} c_0^{4/3})$. \footnote{Previously \cite{blayer08, skoric08, tardos03, laarhoven11}, it was common to parametrize the offset $\delta$ as $\delta = 1/(d_{\delta} c_0)$. However, Laarhoven and De Weger \cite{laarhoven11} showed that to get an optimal codelength, $\delta$ should scale as $c_0^{-4/3}$ rather than $c_0^{-1}$. Therefore we now use $\delta = 1/(d_{\delta} c_0^{4/3})$, with $d_{\delta}$ converging to a non-zero constant for asymptotically large $c_0$.} 
  \end{enumerate}
  \item \textbf{Codeword generation} \\
	For each position $1 \leq i \leq \ell$:
  \begin{enumerate}
    \item Select $p_i \in [\delta, 1 - \delta]$ from the distribution function $F_{\delta}(p)$ defined in \eqref{dist1}.      
    \item For each user $j \in U$, generate the $i$th entry of the codeword of user $j$ according to $\pr(X_{j,i} = 1) = p_i$ and $\pr(X_{j,i} = 0) = 1 - p_i$.
  \end{enumerate}
	\item \textbf{Distribution of codewords}\\
	  Send to each user $j \in U$ their codeword $\vec{X}_j = (X_{j,1}, \ldots, X_{j,\ell})$, embedded as a watermark in the content.
	\item \textbf{Detection of pirate output} \\
	  Detect the pirate output, and extract the watermark $\vec{y} = (y_1, \ldots, y_{\ell})$.
  \item \textbf{Accusation phase} \\
	For each user $j \in U$:
  \begin{enumerate}
    \item For each position $1 \leq i \leq \ell$, calculate the user's score $S_{j,i}$ for this position according to:
      \begin{align}
      S_{j,i} = \begin{cases}
        +\sqrt{(1 - p_i)/p_i} & \text{if $X_{j,i} = 1, y_i = 1$}, \\
        -\sqrt{(1 - p_i)/p_i} & \text{if $X_{j,i} = 1, y_i = 0$}, \\
		-\sqrt{p_i/(1 - p_i)} & \text{if $X_{j,i} = 0, y_i = 1$}, \\
        +\sqrt{p_i/(1 - p_i)} & \text{if $X_{j,i} = 0, y_i = 0$}.
        \end{cases} \label{scoresym}
      \end{align}
    \item Calculate the user's total score $S_j(\ell) = \sum_{i = 1}^{\ell} S_{j,i}$. 
	\item User $j$ is accused (i.e.\ $j \in \hat{C}$) iff $S_j(\ell) > Z$.
  \end{enumerate}
\end{enumerate}

\subsection{Soundness}
\label{sub:Preliminaries-Soundness}

For the above construction, one can prove soundness and static completeness, provided the constants $d_{\ell}, d_z, d_{\delta}$ satisfy certain requirements. For soundness, Laarhoven and De Weger~\cite{laarhoven11} proved the following lemma. Here $h(x) = (e^x - 1 - x)/x^2$, which is a strictly increasing function from $(0,\infty)$ to $(\frac{1}{2}, \infty)$.

\begin{lemma} \cite[Lemma 1]{laarhoven11} \label{lem:Soundness}
Let the Tardos scheme be constructed as in Section~\ref{sub:Preliminaries-Construction}. Let $j$ be some arbitrary innocent user, and let $a > 0$. Then
\begin{align*}
\expn\left(e^{a S_j(\ell) c_0^{-1}}\right) \leq \left(\frac{\eps_1}{n}\right)^{-a \lambda_a d_{\ell}},
\end{align*}
where $\lambda_a = a h(a \sqrt{d_{\delta}} c_0^{-1/3})$.
\end{lemma}

Now if the following condition of soundness is satisfied,
\begin{align*}
\exists \ a > 0: \quad a \left(d_z - \lambda_a d_{\ell}\right) \geq 1, \tag{S} \label{eq:S}
\end{align*}
then using the Markov inequality and Lemma~\ref{lem:Soundness} with this $a$, for innocent users $j$ we get
\begin{align*}
\pr(j \in \hat{C}) & \leq \pr(S_j(\ell) > Z) = \pr(e^{a S_j(\ell) c_0^{-1}} > e^{a Z c_0^{-1}}) \\ 
& \leq \frac{\expn\left(e^{a S_j(\ell) c_0^{-1}}\right)}{e^{a Z c_0^{-1}}} \leq \left(\frac{\eps_1}{n}\right)^{a (d_z - \lambda_a d_{\ell})} \leq \frac{\eps_1}{n}.
\end{align*}
So the probability that no innocent user is accused is at least $(1 - \frac{\eps_1}{n})^n \geq 1 - \eps_1$, as was also shown by Laarhoven and De Weger~\cite[Theorem 3]{laarhoven11}.

\subsection{Static completeness}
\label{sub:Preliminaries-Completeness}

To prove static completeness, Laarhoven and De Weger~\cite{laarhoven11} used the following lemma. Below, and throughout the rest of this paper, $S(\ell) = \sum_{j \in C} S_j(\ell)$ represents the total coalition score, i.e., the sum of the scores of all pirates $j \in C$.

\begin{lemma} \cite[Lemma 2]{laarhoven11} \label{lem:Completeness}
Let the Tardos scheme be constructed as in Section~\ref{sub:Preliminaries-Construction}, and let $b > 0$. Then
\begin{align*}
\expn\left(e^{-b S(\ell) c_0^{-5/3}}\right) \leq \left(\frac{\eps_1}{n}\right)^{b \lambda_b d_{\ell} c_0^{1/3}},
\end{align*}
where $\lambda_b = \frac{2}{\pi} - \frac{4}{d_{\delta} \pi} c_0^{-1/3} - b h(b \sqrt{d_{\delta}}) c_0^{-2/3}$.
\end{lemma}

If the following condition of completeness is satisfied, 
\begin{align*}
\exists \ b > 0: \quad b (\lambda_b d_{\ell} - d_z) \geq \eta c_0^{-1/3}, \tag{C} \label{eq:C}
\end{align*}
then using the pigeonhole principle, the Markov inequality and Lemma~\ref{lem:Completeness} with this $b$ we get 
\begin{align*}
\pr(C \cap \hat{C} = \emptyset) & \leq \pr(S(\ell) < c_0 Z) \leq \frac{\expn\left(e^{-b S(\ell) c_0^{-5/3}}\right)}{e^{-b Z c_0^{-2/3}}} \\
 & \leq \left(\frac{\eps_1}{n}\right)^{b (\lambda_b d_{\ell} - d_z) c_0^{1/3}} \leq \left(\frac{\eps_1}{n}\right)^{\eta} = \eps_2.
\end{align*}
So static completeness follows from Lemma~\ref{lem:Completeness} and condition~\eqref{eq:C}, as was also shown by Laarhoven and De Weger~\cite[Theorem 4]{laarhoven11}.

\subsection{Codelengths}
\label{sub:Preliminaries-Optimization}

Blayer and Tassa~\cite{blayer08}, and subsequently Laarhoven and De Weger~\cite{laarhoven11}, gave a detailed analysis to go from requirements \eqref{eq:S} and \eqref{eq:C} to the optimal set of parameters that satisfies the constraints and minimizes $d_{\ell}$. Recall that $\ell = d_{\ell} c_0^2 \ln(n/\eps_1)$, so a smaller $d_{\ell}$ gives shorter codelengths, whereas the parameters $d_z$ and $d_{\delta}$ affect only $Z$ and $\delta$, which have no influence on the efficiency of the scheme. In the end, the following result was obtained.

\begin{lemma} \cite[Theorem 6]{laarhoven11} \label{lem:FirstOrder} 
Let $\gamma = \left(\frac{2}{3\pi} \right)^{2/3} \approx 0.36$. The asymptotically optimal value for $d_{\ell}$ is
\begin{align*}
  d_{\ell} &= \frac{\pi^2}{2} + O(c_0^{-1/3}),
\end{align*}
the associated values for $d_z$ and $d_{\delta}$ are
\begin{align*}
  d_z = \pi + O(c_0^{-1/3}), & \quad d_{\delta} = \frac{4}{\gamma} - O\left(\frac{\eta}{\ln c_0}\right),
\end{align*}
and the corresponding values for $a,b,\lambda_a, \lambda_b$ are
\begin{align*}
  a = \frac{2}{\pi} - O(c_0^{-1/3}), & \quad b = \frac{\ln c_0}{9 \pi \gamma} - O\left(\ln\left(\frac{\ln c_0}{\eta}\right)\right), \\
	\lambda_a = \frac{1}{\pi} + O(c_0^{-1/3}), & \quad \lambda_b = \frac{2}{\pi} - O(c_0^{-1/3}).
\end{align*}
\end{lemma}

A direct consequence of Lemma~\ref{lem:FirstOrder} is the following, which gives the asymptotically optimal scheme parameters for $c_0 \to \infty$.

\begin{corollary} \cite[Corollary 1]{laarhoven11} \label{cor:Asymptotics}
The construction from Section \ref{sub:Preliminaries-Construction} gives an $\eps_1$-sound and static $(\eps_2,c_0)$-complete scheme with asymptotic scheme parameters
\begin{align*}
\ell \to \frac{\pi^2}{2} c_0^2 \ln(n/\eps_1), \quad Z \to \pi c_0 \ln(n/\eps_1), \quad \delta \to \frac{\gamma}{4} c_0^{-4/3}.
\end{align*}
\end{corollary}

For further details on the optimal first order constants, see Laarhoven and De Weger~\cite{laarhoven11}.

\subsection{Example}
\label{sub:Preliminaries-Example}

For the next few sections, we will use a running example to compare the codelengths of the several schemes. Let the scheme parameters be given by $c_0 = 25$ pirates, $n = 10^6$ users, and error probabilities $\eps_1 = \eps_2 = 10^{-3}$. Then $\eta = \frac{1}{3}$, and the optimal values of $d_{\ell}, d_z, d_{\delta}$ can be calculated numerically as
\begin{align*}
  d_{\ell} = 8.46, \quad d_z = 4.53, \quad d_{\delta} = 14.36.
\end{align*}
This leads to the scheme parameters
\begin{align*}
  \ell = 109\,585, \quad Z = 2345, \quad \delta = 5.09 \cdot 10^{-4}.
\end{align*}
So using these scheme parameters, we know that after $109\,585$ symbols, with probability at least $0.999$ there are no false accusations (regardless of the actual coalition size $c$), and with probability at least $0.999$ at least one pirate is accused if the actual coalition size $c$ does not exceed the bound on the coalition size $c_0 = 25$. In Fig.~\ref{fig:Fig1} we show simulation results for these parameters, with $c = c_0 = 25$. The curves in the figure are the pirate scores $S_j(i)$ for each pirate $j \in C$, while the shaded area is bounded from above by the highest score of an innocent user, and bounded from below by the lowest score of an innocent user in this simulation. In Fig.~\ref{fig:Fig1a} we simulated pirates using the interleaving attack (i.e.\ for each position, they choose a random pirate and output his symbol), and in Fig.~\ref{fig:Fig1b} they used the scapegoat strategy (i.e.\ one pirate, the scapegoat, always outputs his symbol, until he is caught and another pirate is picked as the scapegoat). With the scapegoat strategy, only one pirate is caught, while using the interleaving attack leads to many accused pirates.

\begin{figure}[!t]
\centering
\subfloat[][Interleaving attack]{\includegraphics[width=\columnwidth]{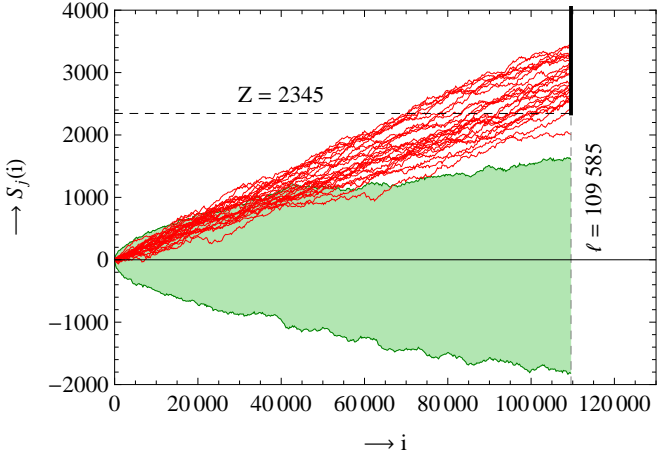} \label{fig:Fig1a}} \\
\subfloat[][Scapegoat strategy]{\includegraphics[width=\columnwidth]{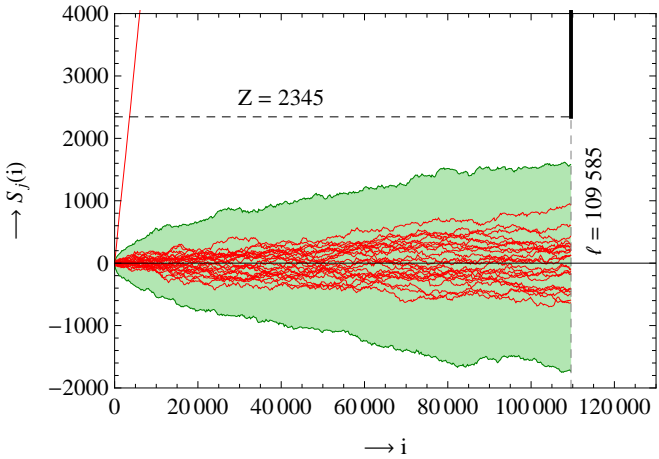} \label{fig:Fig1b}}
\caption{Simulations of the Tardos scheme, with $c = c_0 = 25$ colluders, $n = 10^6$ users, and error probabilities $\eps_1 = \eps_2 = 10^{-3}$. The green, shaded area corresponds to the range of innocent user scores, the red lines correspond to pirate scores, and the dashed lines correspond to the threshold $Z$ and codelength $\ell$. In Fig.~\ref{fig:Fig1a} the pirates used the interleaving attack, whereas in Fig.~\ref{fig:Fig1b} they used the scapegoat strategy. In both cases, the total coalition score $S(\ell)$ at time $\ell$ is approximately $72\,000$, but while in the first case the score is evenly divided among the pirates, in the second case one pirate takes all the blame.}
\label{fig:Fig1}
\end{figure}

\section{The dynamic Tardos scheme}
\label{sec:DynamicTardos}

Let us now explain how we create a dynamic scheme from the static Tardos scheme, such that with high probability we catch all colluders, instead of at least one colluder. The change we make is the following. Instead of only comparing the cumulative user scores to $Z$ after $\ell$ symbols, we now compare the scores to $Z$ after every single position $i$. If a user's score exceeds $Z$ at any point in time, he is disconnected immediately and can no longer access the content. His score is then necessarily between $Z$ and $\tilde{Z} := Z + \sqrt{d_{\delta}} c_0^{2/3} > Z + \max_{p_i,X_{j,i},y_i} S_{j,i}$. The other parts of the construction remain the same, except for the values of $d_{\ell}, d_z, d_{\delta}$, which now have to be chosen differently.

\subsection{Construction}
\label{sub:DynamicTardos-Construction}

The scheme again depends on three constants $d_{\ell}, d_z, d_{\delta}$. We will show in Sections~\ref{sub:DynamicTardos-Soundness} and \ref{sub:DynamicTardos-Completeness} that if certain requirements on these constants are satisfied, we can prove soundness and dynamic completeness. Below we say a user is active if he has not yet been disconnected from the scheme. As mentioned before, we assume that the pirates always output some watermarked data, unless all of the pirates are disconnected. In that case, the traitor tracing scheme terminates. 

\begin{enumerate}
  \item \textbf{Initialization phase}
  \begin{enumerate}
    \item Take the codelength as $\ell = d_{\ell} c_0^2 \ln(n/\eps_1)$.
		\item Take the threshold as $Z = d_z c_0 \ln(n/\eps_1)$.
		\item Take the cutoff parameter as $\delta = 1/(d_{\delta} c_0^{4/3})$.
		\item Set initial user scores at $S_j(0) = 0$.
  \end{enumerate}
  \item \textbf{Codeword generation} \\
	For each position $1 \leq i \leq \ell$:
  \begin{enumerate}
    \item Select $p_i \in [\delta, 1 - \delta]$ from $F_{\delta}(p)$ defined in \eqref{dist1}.
    \item Generate $X_{j,i} \in \{0,1\}$ using $\pr(X_{j,i} = 1) = p_i$.
  \end{enumerate}
  \item \textbf{Distribution/Detection/Accusation} \\
  For each position $1 \leq i \leq \ell$:
	\begin{enumerate}
	  \item Send to each active user $j$ symbol $X_{j,i}$.
		\item Detect the pirate output $y_i$. \\
		      (If there is no pirate output, terminate.)
		\item Calculate scores $S_{j,i}$ using \eqref{scoresym}.
    \item For active users $j$, set $S_j(i) = S_j(i-1) + S_{j,i}$. \\
		      (For inactive users $j$, set $S_j(i) = S_j(i-1)$.)
		\item Disconnect all active users $j$ with $S_j(i) > Z$.
	\end{enumerate}
\end{enumerate}

In the construction above, we separated the codeword generation from the distribution, detection and accusation. These phases can also be merged by generating $p_i$ and $X_{j,i}$ once we need them. However, we present the scheme as above to emphasize the fact that these phases can indeed be executed sequentially instead of simultaneously, and that the codeword generation can thus be done before the traitor tracing process begins. 

\subsection{Soundness}
\label{sub:DynamicTardos-Soundness}

For the dynamic Tardos scheme as given above, we can prove the following result regarding soundness.

\begin{theorem} \label{thm:PDS-DT-Soundness}
Consider the dynamic Tardos scheme in Section~\ref{sub:DynamicTardos-Construction}. If the following condition is satisfied,
\begin{align*}
\exists \ a > 0: \quad a (d_z - \lambda_a d_{\ell}) \geq 1 + \frac{\ln(2)}{\ln(n/\eps_1)}, \tag{S'} \label{eq:S'}
\end{align*}
then the scheme is $\eps_1$-sound.
\end{theorem}

To prove the theorem, we first prove a relative upper bound on the probability that a single innocent user is accused and disconnected. This bound relates the error probability in the dynamic Tardos scheme to the probability that the user score at time $\ell$ is above $Z$. We then use the proof of the original Tardos scheme to get an absolute upper bound on the soundness error probability, and to prove Theorem~\ref{thm:PDS-DT-Soundness}. Since the relative upper bound gives us an extra factor $2$, and since the terms in \eqref{eq:S'} appear as exponents in the proof, we get an additional term $\ln(2)/\ln(n/\eps_1)$ compared to \eqref{eq:S}. Note that this term is small for reasonable values of $n$ and $\eps_1$, so this only has a small impact on the right hand side of \eqref{eq:S'}, compared to \eqref{eq:S}. 

In the following we write $\tilde{S}_j(i) = \sum_{k=1}^i S_{jk}$ for the \textit{extended} user score. If user $j$ is still active at time $i$, then $\tilde{S}_j(i) = S_j(i)$. But whereas $S_j(i)$ does not change anymore once user $j$ is disconnected, the score $\tilde{S}_j(i)$ does change on every position, even if the user has already been disconnected. The score $\tilde{S}_j$ then calculates the user's score as if he had not been disconnected. Similarly, we write $\tilde{S}(i) = \sum_{j \in C} \tilde{S}_j(i)$ for coalitions $C$. Note that if the last pirate is disconnected at position $i_0 < \ell$, then $S_{j,i}$ and $S_j(i)$ are not defined for $i_0 < i \leq \ell$.

\begin{lemma} \label{lem:DT1}
Let $j \in U$ be an arbitrary innocent user, let $C \subseteq U \setminus \{j\}$ be a pirate coalition and let $\rho$ be some pirate strategy employed by this coalition. Then
\begin{align*}
\pr(j \in \hat{C}) = P(S_j(\ell) > Z) \leq 2 \cdot \pr\left(\tilde{S}_j(\ell) > Z\right).
\end{align*}
\end{lemma}

\begin{proof}
Let us define events $A$ and $B$ as
\begin{align*}
A &:= \{j \in \hat{C}\} = \{S_j(\ell) > Z\} = \bigcup_{i=1}^{\ell} \{\tilde{S}_j(i) > Z\}, \\
B &:= \{\tilde{S}_j(\ell) > Z\}.
\end{align*}
We trivially have $\pr(A \mid B) = 1$. For $\pr(B \mid A)$, note that under the assumption that $A$ holds, the process $\{\tilde{S}_j(i)\}_{i = i_0}^{\infty}$ starting at position $i_0 = \min \{i: S_j(i) > Z\} \leq \ell$ describes a symmetric random walk with no drift. So we then have $\pr(\tilde{S}_j(\ell) \geq \tilde{S}_j(i_0)) = 1/2$, and since $S_j(i_0) > Z$ it follows that $\pr(B \mid A) \geq 1/2$. Finally we apply Bayes' Theorem to $A$ and $B$ to get 
\begin{align*}
P(A) = \frac{P(A \mid B)}{P(B \mid A)} \cdot P(B) \leq 2 \cdot P(B).
\end{align*}
This completes the proof. 
\end{proof}

\begin{proof}[Proof of Theorem \ref{thm:PDS-DT-Soundness}]
First, we remark that the distribution of $\tilde{S}_j(\ell)$ is the same as the distribution of the scores $S_j(\ell)$ in the original Tardos scheme, for the same parameters $\ell, Z, \delta$. From the Markov inequality, Lemma~\ref{lem:Soundness} and condition~\eqref{eq:S'} it thus follows that
\begin{align*}
\pr(\tilde{S}_j\left(\ell) > Z\right) \leq \frac{\expn\left(e^{a \tilde{S}_j(\ell) c_0^{-1}}\right)}{e^{a Z c_0^{-1}}} \leq \left(\frac{\eps_1}{n}\right)^{a (d_z - \lambda_a d_{\ell})} \leq \frac{\eps_1}{2n}.
\end{align*}
Using Lemma~\ref{lem:DT1} the result follows.
\end{proof}

\subsection{Dynamic completeness}
\label{sub:DynamicTardos-Completeness}

With the dynamic Tardos scheme, we get the following result regarding dynamic completeness. Recall that here we require that \textit{all} pirates are caught, instead of at least one, as was the case in the original Tardos scheme. 

\begin{theorem} \label{thm:PDS-DT-Completeness}
Consider the dynamic Tardos scheme in Section~\ref{sub:DynamicTardos-Construction}. If the following condition is satisfied,
\begin{align*}
\exists \ b > 0: b \left(\lambda_b d_{\ell} - d_z\right) \geq \left(\eta + \frac{\ln(2) + b \sqrt{d_{\delta}}}{\ln(n/\eps_1)}\right) c_0^{-1/3}, \tag{C'} \label{eq:C'}
\end{align*}
then the scheme is dynamic $(\eps_2,c_0)$-complete.
\end{theorem}

Similar to the proof of soundness, we prove dynamic completeness by relating the error probability to the static completeness error probability of the static Tardos scheme described in Section~\ref{sec:Preliminaries}. Then we use the results from the static scheme to complete the proof. We again see a factor $2$ in the relative upper bound in Lemma~\ref{lem:DT-Comp}, which again comes from a random walk argument, and which explains the additional term $\ln(2)/\ln(n/\eps_1)$ in~\eqref{eq:C'}. The other term $b \sqrt{d_{\delta}}/\ln(n/\eps_1)$ is a consequence of using $\tilde{Z}$ instead of $Z$ in the proofs. Note that these two terms are generally small, compared to the term $\eta$. 

\begin{lemma} \label{lem:DT-Comp}
Let $C$ be a coalition of size at most $c_0$, and let $\rho$ be any pirate strategy employed by this coalition. Then
\begin{align*}
\pr\left(C \not\subseteq \hat{C}\right) \leq 2 \cdot \pr\left(\tilde{S}(\ell) < c_0\tilde{Z}\right).
\end{align*}
\end{lemma}

\begin{proof}
First we remark that $\pr(\tilde{S}(\ell) < c_0\tilde{Z} \mid C \not\subseteq \hat{C}) \geq 1/2$. In other words, if not all pirates are caught by the end, the total extended coalition score will be below $c_0 \tilde{Z}$ with probability at least $1/2$. This is because if $C \not\subseteq \hat{C}$, then $S(\ell) < c_0\tilde{Z}$, and since $\tilde{S}(\ell) - S(\ell) = R(\ell)$ is a symmetric, unbiased random walk, with probability at least $1/2$ we have $R(\ell) < 0$ and as a consequence $\tilde{S}(\ell) < c_0\tilde{Z}$. Next, we use the definition of conditional probabilities to get
\begin{align*}
\pr(C \not\subseteq \hat{C}) \leq 2 \cdot \pr(C \not\subseteq \hat{C}) \cdot \pr(\tilde{S}(\ell) < c_0 \tilde{Z} \mid C \not\subseteq \hat{C}) \\ = 2 \cdot \pr\left(\tilde{S}(\ell) < c_0 \tilde{Z}, C \not\subseteq \hat{C}\right) \leq 2 \cdot \pr(\tilde{S}(\ell) < c_0\tilde{Z}).
\end{align*}
This proves the result.
\end{proof}

\begin{proof}[Proof of Theorem \ref{thm:PDS-DT-Completeness}]
First, note that in the dynamic Tardos scheme, the only extra information pirates receive compared to the static Tardos scheme is the fact whether some of them are disconnected. This information is certainly covered by the information contained in the previous values of $p_i$; if pirates receive $p_1, \ldots, p_{i-1}$, then they can calculate their current scores themselves and calculate whether they would have been disconnected or not. Also note that $\tilde{S}(\ell)$ behaves the same as $S(\ell)$ in the static Tardos scheme, where the total coalition score is calculated for all pirates and all positions, regardless of whether they contributed on that position or not. So if we can prove that even in the static Tardos scheme, and even if coalitions get information about the previous values of $p_i$ (for which $y_i$ was already determined), the probability of keeping the coalition score $S(\ell)$ below $c_0\tilde{Z}$ is bounded by $\eps_2/2$, then it follows that also $\pr(\tilde{S}(\ell) < c_0 \tilde{Z}) \leq \eps_2/2$.

For the static Tardos scheme, note that the proof method for the completeness property does not rely on the other values of $p_i$ being secret. In fact, $p_i$ and $p_{i'}$ are independent for $i \neq i'$. The only assumption that is used in that proof is that the Marking Assumption applies, which does apply here, and that the \textit{current} value $p_i$ is hidden before $y_i$ is generated. So here we can also use the proof method of the static Tardos scheme. From the Markov inequality, Lemma~\ref{lem:Completeness} and condition~\eqref{eq:C'}, it thus follows that
\begin{align*}
\pr\left(\tilde{S}(\ell) < c_0\tilde{Z}\right) 
 & \leq \frac{\expn\left(e^{-b \tilde{S}(\ell) c_0^{-5/3}}\right)}{e^{-b (Z + \sqrt{d_{\delta}} c_0^{2/3}) c_0^{-2/3}}} \\ 
 & \leq \left(\frac{\eps_1}{n}\right)^{b \left(\lambda_b d_{\ell} - d_z - \frac{\sqrt{d_{\delta}}}{\ln(n/\eps_1)} c_0^{-1/3}\right) c_0^{1/3}} \\ 
 & \leq \left(\frac{\eps_1}{n}\right)^{\eta + \frac{\ln 2}{\ln(n/\eps_1)}} = \frac{\eps_2}{2}.
\end{align*}
Using Lemma~\ref{lem:DT-Comp} the result then follows.
\end{proof}

\subsection{Codelengths}
\label{sub:DynamicTardos-Optimization}

The requirements \eqref{eq:S'} and \eqref{eq:C'} are only slightly different from requirements \eqref{eq:S} and \eqref{eq:C}. For asymptotically large $c_0$, these differences even disappear, and the optimal asymptotic codelength is the same as in the static Tardos scheme. In Fig.~\ref{fig:Fig2} we show the optimal values of $d_{\ell}$ in the dynamic Tardos scheme for $\eta = 1$ and $\eta = 0.01$. The different curves correspond to different values of $n/\eps_1$, ranging from $n/\eps_1 = 10^3$ (the highest values of $d_{\ell}$) to $n/\eps_1 = 10^{15}$ (the lowest values of $d_{\ell}$).

\begin{figure}[!t]
\centering
\includegraphics[width=\columnwidth]{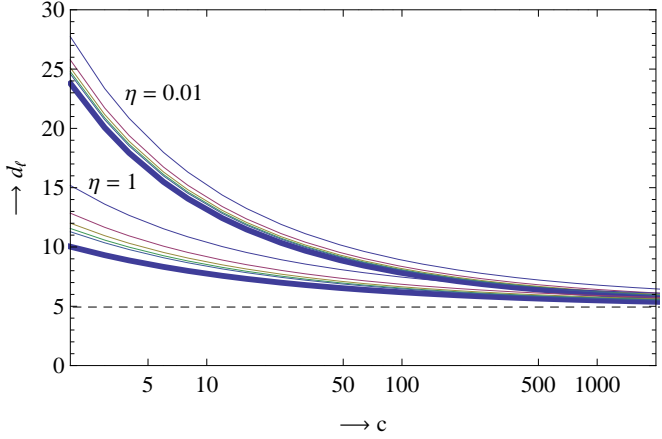}
\caption{Optimal values of $d_{\ell}$ in the dynamic Tardos scheme. The dotted line corresponds to the asymptotic optimal value $d_{\ell} = \frac{\pi^2}{2} \approx 4.93$. The bold curves show the values of $d_{\ell}$ in the static Tardos scheme for $\eta = 1$ (top) and $\eta = 0.01$ (bottom) respectively. The five curves slightly above each of the bold curves show the optimal values of $d_{\ell}$ in the dynamic Tardos scheme for $n/\eps_1 = 10^{3k}$, for $k = 1$ up to $5$. Higher values of $k$ correspond to lower values of $d_{\ell}$.}
\label{fig:Fig2}
\end{figure}

Note that these values of $d_{\ell}$ correspond to the theoretical codelengths such that with probability at least $1 - \eps_1$, by time $\ell$ all of the pirates have been disconnected. This does not mean that the last pirate is likely to be caught \textit{exactly at} time $\ell$; this means that he is likely to be caught \textit{before or at} time $\ell$. So in practice the number of symbols needed to disconnect all traitors may very well be below this theoretical codelength $\ell$, and may even decrease compared to the static Tardos scheme. 

Furthermore, if the coalition size is not known, then one generally uses a traitor tracing scheme that is resistant against up to $c_0 > c$ colluders. \v{S}kori\'{c} et al.~\cite{skoric08} showed that in the Tardos scheme, the total coalition score $S(i) = \sum_{j \in C} S_j(i)$ always increases linearly in $i$ with approximately the same slope, regardless of the actual coalition size $c$ or the employed pirate strategy $\rho$. More precisely, the score $S(i)$ behaves as $S(i) \approx i \tilde{\mu}$, with $\tilde{\mu} \approx \frac{2}{\pi}$ only slightly depending on the coalition size $c$ and the pirate strategy $\rho$. Since one chooses $\ell$ and $Z$ such that $S(\ell) \approx \ell \tilde{\mu} \approx c_0 Z$, it follows that $S(\frac{c}{c_0} \ell) \approx c Z$. In other words, to catch a coalition of size $c \leq c_0$, the expected number of symbols needed is approximately $\ell = O(\frac{c}{c_0} \ell) = O(c c_0 \ln(n/\eps_1))$. So compared to the static Tardos scheme, where the codelength is fixed in advance at $O(c_0^2 \ln(n/\eps_1))$, the codelength is reduced by a factor $\frac{c}{c_0}$. In particular, small coalitions of few pirates are generally caught up to $O(c_0)$ times faster, for $c_0 \gg c$. 

\subsection{Example}
\label{sub:DynamicTardos-Example}

Let the scheme parameters be the same as in Section~\ref{sub:Preliminaries-Example}, i.e., $c_0 = 25$, $n = 10^6$ and $\eps_1 = \eps_2 = 10^{-3}$, so that $\eta = \frac{1}{3}$. The optimal values of $d_{\ell}, d_z, d_{\delta}$ satisfying \eqref{eq:S'} and \eqref{eq:C'} can be calculated numerically as
\begin{align*}
  d_{\ell} = 9.00, \quad d_z = 4.73, \quad d_{\delta} = 13.44
\end{align*}
This leads to the scheme parameters
\begin{align*}
  \ell = 116\,561, \quad Z = 2448, \quad \delta = 1.02 \cdot 10^{-3}
\end{align*}
In Fig.~\ref{fig:Fig3} we show some simulation results for these parameters, with the actual coalition also consisting of $c = c_0 = 25$ colluders. In Fig.~\ref{fig:Fig3a} the pirates used the interleaving attack, and in Fig.~\ref{fig:Fig3b} they used the scapegoat strategy. In both cases, the whole coalition is caught well before we reach $\ell$ symbols.

\begin{figure}[!t]
\centering
\subfloat[][Interleaving attack]{\includegraphics[width=\columnwidth]{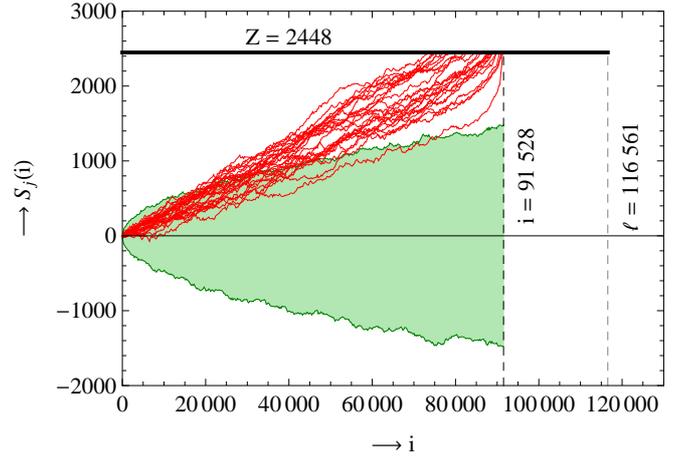} \label{fig:Fig3a}} \\
\subfloat[][Scapegoat strategy]{\includegraphics[width=\columnwidth]{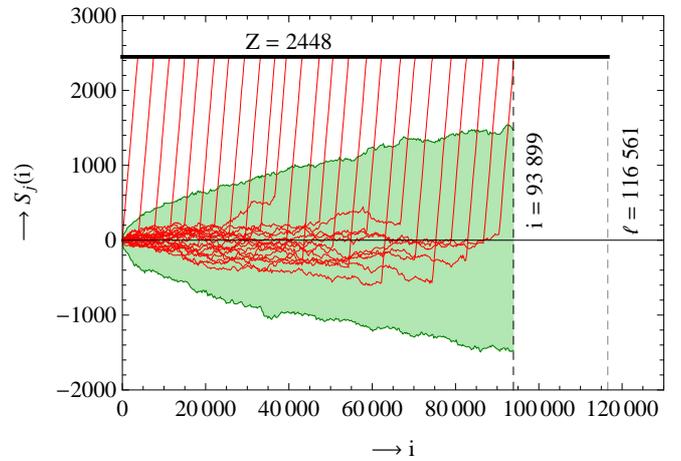} \label{fig:Fig3b}}
\caption{Simulations of the dynamic Tardos scheme, with the same parameters $c$, $c_0$, $n$, $\eps_1$ and $\eps_2$ as in Fig.~\ref{fig:Fig1}. Users are now disconnected as soon as their scores exceed the threshold $Z$, i.e., as soon as their corresponding score curves cross the bold horizontal line. In both cases, after less than $95\,000$ symbols all pirates have been caught, which is less than the theoretical codelength $\ell = 116\,561$, and less than the codelength of the static Tardos scheme with the same parameters, $\ell = 109\,585$.}
\label{fig:Fig3}
\end{figure}

\section{The weakly dynamic Tardos scheme}
\label{sec:SemiDynamicTardos}

In the dynamic Tardos scheme, we need to disconnect users as soon as their scores exceed the threshold $Z$. In some scenarios this may not be possible. For example, the pirates may transmit each symbol with a delay. 

We call a traitor tracing scheme weakly dynamic if $B$ ($B \geq 1$) symbols are distributed during the delay between the original broadcast and the corresponding pirate output. Observe that the dynamic schemes presented in \cite{berkman01, fiat01, roelse11, tassa05} are not weakly dynamic schemes, as these schemes use the value of each symbol to adapt the distribution of the next symbols (i.e. $B = 0$ for these schemes). 

In this section we present two weakly dynamic schemes based on the dynamic Tardos scheme.
First, in Section~\ref{sub:SemiDynamicTardos-First} we present a scheme that achieves a codelength of at most $\ell = d_{\ell} c_0^2 \ln(n/\eps_1) + B c_0$, where $d_{\ell}$ is the same as in the dynamic Tardos scheme for the same parameters. For small values of $B$, this means that with codelength which is only slightly higher than in the dynamic Tardos scheme, we can also catch all pirates in a weakly dynamic setting. Then, in Section~\ref{sub:SemiDynamicTardos-Second} we present a scheme that achieves a codelength of $\ell = d_{\ell,B} c_0^2 \ln(n/\eps_1)$, where $d_{\ell,B}$ increases with $B$. Since a small increase in $d_{\ell}$ can already lead to a big increase in the codelength, the second scheme generally has a larger codelength than the first scheme. 

\subsection{First scheme: $\ell = d_{\ell} c_0^2 \ln(n/\eps_1) + B c_0$}
\label{sub:SemiDynamicTardos-First}

The first scheme is based on the following modification to the accusation algorithm of the dynamic Tardos scheme. Suppose a user's score exceeds $Z$ after $i_0$ positions. At position $i_0$ we now disconnect this user. Since this user may have contributed to the next $B$ symbols of the pirate output $\vec{y}$, we disregard the following $B$ `contaminated' positions of the watermark, and do not update the scores for positions $i \in \{i_0 + 1, \ldots, i_0 + B\}$. After those positions we continue the traitor tracing process as in the dynamic Tardos scheme, and we repeat the above procedure each time a user's score exceeds $Z$. 

With this modification, the traitor tracing process on those positions that were used for calculating scores is identical to the traitor tracing process of the dynamic Tardos scheme. We can therefore use the analysis from Section~\ref{sec:DynamicTardos} and conclude that with at most $d_{\ell} c_0^2 \ln(n/\eps_1)$ positions for which we calculate scores, we can catch any coalition of size $c \leq c_0$. Since we disregarded at most $B c_0$ positions, the pirate broadcast will not last longer than $\ell = d_{\ell} c_0^2 \ln(n/\eps_1) + B c_0$ positions in total, where $d_{\ell}$, $d_z$ and $d_{\delta}$ are as in the dynamic Tardos scheme for the same parameters. This means that with at most $B c_0$ more symbols than in the dynamic Tardos scheme, we can also catch coalitions in this weakly dynamic traitor tracing setting. 

\subsection{Second scheme: $\ell = d_{\ell,B} c_0^2 \ln(n/\eps_1)$}
\label{sub:SemiDynamicTardos-Second}

Instead of using $B c_0$ more symbols, we can also try to adjust the analysis of the dynamic Tardos scheme to the weakly dynamic traitor tracing scenario. We can do this by following the proof methods of the dynamic Tardos scheme, and by making one small adjustment. The change we make in the analysis is to use $\tilde{Z}_B := Z + B \sqrt{d_{\delta}} c_0^{2/3} > Z + B \max_p S_{j,i}(p)$ instead of $\tilde{Z} = Z + \sqrt{d_{\delta}} c_0^{2/3}$ as our new upper bound for the scores of users in the proofs. This results in the following, slightly different condition for dynamic completeness:
\begin{align*}
\exists \ b > 0: b (\lambda_b d_{\ell} - d_z) \geq \left(\eta + \frac{\ln(2) + B b \sqrt{d_{\delta}}}{\ln(n/\eps_1)}\right) c_0^{-1/3}. \tag{C''} \label{eq:C''}
\end{align*}
If some parameters $d_{\ell,B}$, $d_{z,B}$, $d_{\delta,B}$ satisfy \eqref{eq:S'} and \eqref{eq:C''}, then using these constants as our scheme parameters, we obtain a $\eps_1$-sound and dynamic $(\eps_2, c_0)$-complete scheme with a codelength of $\ell = d_{\ell,B} c_0^2 \ln(n/\eps_1)$. In Fig.~\ref{fig:Fig4} we show the values of $d_{\ell,B}$ for the parameters $n = 10^6$, $\eps_1 = \eps_2 = 10^{-3}$, and $\eta = \frac{1}{3}$, for several values of $B$. As the value of $B$ increases, the values of $d_{\ell,B}$ increase as well.

\begin{figure}[!t]
\centering
\includegraphics[width=\columnwidth]{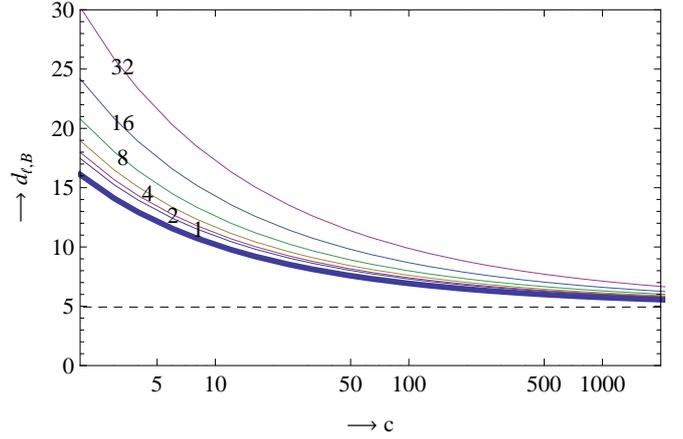}
\caption{Optimal values of $d_{\ell,B}$ in the weakly dynamic Tardos scheme from Section~\ref{sub:SemiDynamicTardos-Second}, for the parameters $n = 10^6$, $\eps_1 = \eps_2 = 10^{-3}$, and $\eta = \frac{1}{3}$. The bold curve corresponds to the values of $d_{\ell}$ in the static Tardos scheme with the same parameters, while the six curves above this curve correspond to the optimal values of $d_{\ell,B}$ for $B = 1, 2, 4, 8, 16, 32$ respectively. The dotted line corresponds to the asymptotic optimal value $d_{\ell} = \frac{\pi^2}{2} \approx 4.93$. For $B = 1$ we get exactly the codelengths of the dynamic Tardos scheme.}
\label{fig:Fig4}
\end{figure}

\subsection{Example}
\label{sub:SemiDynamicTardos-Example}

As before, let the scheme parameters be given by $c_0 = 25$, $n = 10^6$ and $\eps_1 = \eps_2 = 10^{-3}$, so that $\eta = \frac{1}{3}$, and let us use $B = 8$. With the first proposed scheme, the codelength increases by $B c_0 = 200$ symbols compared to the dynamic Tardos scheme, giving scheme parameters:
\begin{align*}
  \ell = 116\,761, \quad Z = 2448, \quad \delta = 1.02 \cdot 10^{-3}.
\end{align*}
Using the second scheme, the optimal values of $d_{\ell,B}, d_{z,B}, d_{\delta,B}$ satisfying \eqref{eq:S'} and \eqref{eq:C''} for $B = 8$ can be calculated numerically as
\begin{align*}
  d_{\ell,B} = 10.16, \quad d_{z,B} = 4.94, \quad d_{\delta,B} = 10.07.
\end{align*}
This leads to the scheme parameters
\begin{align*}
  \ell = 131\,587, \quad Z = 2561, \quad \delta = 1.36 \cdot 10^{-3}.
\end{align*}
So in this case, using the first scheme leads to the shortest code. 

\section{The universal Tardos scheme}
\label{sec:UniversalTardos}

In this section we present a dynamic scheme that does not require a sharp upper bound $c_0$ on $c$ as input to guarantee quick detection of pirates. This means that even if we set $c_0 = n$, coalitions of any size are caught quickly. We use the word ``universal'' to indicate this universality with respect to the coalition size: coalitions of any size can be caught efficiently with this scheme. Note that in the (dynamic) Tardos scheme, we used the distribution function $F_{\delta}$ where $\delta = \delta(c_0) = O(c_0^{-4/3})$ depends on $c_0$. Instead, we will use a distribution function $F$ that can be used for all values of $c$, so that we can use the same codewords to catch coalitions of any size. In particular, we will use the first $\ell^{(c)} = O(c^2 \ln(n/\eps_1))$ symbols to catch coalitions of size $c$, for each $c$ between $2$ and $c_0$. We do this in such a way that if a coalition has some unknown size $c$, then after $\ell^{(c)} = O(c^2 \ln(n/\eps_1)$ symbols, the probability of not having caught all members of this coalition is at most $\eps_2$. Since we do this for each value of $c$, we now only need $O(c^2 \ln(n/\eps_1))$ symbols to catch a coalition of a priori unknown size $c$, compared to the $O(c_0^2 \ln(n/\eps_1))$ worst-case codelength of the static and dynamic Tardos schemes, and the $O(c c_0 \ln(n/\eps_1))$ practical codelength of the dynamic Tardos scheme.

The only drawback of this new codeword generation method is that a completely universal distribution function, which is completely efficient for all values of $c$, does not seem to exist. More precisely, the proof of soundness of the Tardos scheme requires the cutoff parameter $\delta$ to be sufficiently large in terms of $c$, whereas for completeness we need that $\delta$ approaches $0$ as $c \to \infty$. Our solution to this problem is the following. For generating the values of $p_i$, we use the standard arcsine distribution function $F$ from Eq.~\eqref{dist2}, with no cutoffs. Then, for each value of $c$, we simply disregard those values $p_i$ that are not between the corresponding cutoff $\delta^{(c)}$ and $1 - \delta^{(c)}$. The fraction of values of $p_i$ that is disregarded can be estimated as follows:
\begin{align*}
1 - \int_{\delta^{(c)}}^{1 - \delta^{(c)}} \! f(p) \, \mathrm{d}p = \frac{4}{\pi} \arcsin \sqrt{\delta^{(c)}} = \frac{4 c^{-2/3}}{\pi \sqrt{d_{\delta}}} + O(c^{-2}).
\end{align*}
So the fraction of disregarded positions is very small and decreases when $c$ increases.

\subsection{Construction}
\label{sub:UniversalTardos-Construction}

The construction now basically consists of running several dynamic Tardos schemes simultaneously with shared codewords. So scheme parameters and scores now have to be calculated for each of these schemes, i.e., for each of the values of $c$. We introduce counters $t^{(c)}$ to keep track of the number of positions that have not been disregarded. For each $c$, we then run a dynamic Tardos scheme using the same code $X$ until $t^{(c)} = \ell^{(c)}$.

\begin{enumerate}
  \item \textbf{Initialization phase} \\
	For each $c \in \{2, \ldots, c_0 = n\}$:
  \begin{enumerate}
    \item Take the codelength as $\ell^{(c)} = d_{\ell}^{(c)} c^2 \ln(n/\eps_1^{(c)})$.
		\item Take the threshold as $Z^{(c)} = d_z^{(c)} c \ln(n/\eps_1^{(c)})$.
		\item Take the cutoff parameter as $\delta^{(c)} = 1/(d_{\delta}^{(c)} c^{4/3})$.
		\item Initialize the user scores at $S_j^{(c)}(0) = 0$.
		\item Initialize the counters $t^{(c)}$ at $t^{(c)}(0) = 0$.
  \end{enumerate}
  \item \textbf{Codeword generation} \\
	For each position $i \geq 1$:
  \begin{enumerate}
    \item Select $p_i \in [0, 1]$ from $F(p)$ as defined in \eqref{dist2}.
    \item Generate $X_{j,i} \in \{0, 1\}$ using $\pr(X_{j,i} = 1) = p_i$.
  \end{enumerate}
  \item \textbf{Distribution/Detection/Accusation} \\
	For each position $i \geq 1$:
	\begin{enumerate}
	  \item Send to each active user $j$ symbol $X_{j,i}$.
		\item Detect the pirate output $y_i$. \\
		      (If there is no pirate output, terminate.)
		\item Calculate scores $S_{j,i}$ using \eqref{scoresym}.
    \item For active users $j$ and values $c$ such that $p_i \in [\delta^{(c)}, 1 - \delta^{(c)}]$, set $S_j^{(c)}(i) = S_j^{(c)}(i-1) + S_{j,i}$. \\
		      (Otherwise set $S_j^{(c)}(i) = S_j^{(c)}(i-1)$.)
		\item For values of $c$ such that $p_i \in [\delta^{(c)}, 1 - \delta^{(c)}]$, set $t^{(c)}(i) = t^{(c)}(i-1) + 1$. \\
		      (Otherwise set $t^{(c)}(i) = t^{(c)}(i-1)$.)
		\item Disconnect all active users $j$ with $S_j^{(c)}(i) > Z^{(c)}$ and $t^{(c)}(i) \leq \ell^{(c)}$ for some $c$.
	\end{enumerate}
\end{enumerate}

As was already mentioned in Section~\ref{sub:DynamicTardos-Construction}, if desired the codeword generation can be merged with the distribution/detection/accusation phase. This depends on the scenario and the exact implementation of the scheme. 

Also note that several variations can be made to the above construction, to deal with specific situations. One could easily replace $c_0 = n$ by a smaller value of $c_0$ to restrict the amount of memory needed, if a sharper upper bound on $c$ is known. And of course, we may also choose to draw values $p_i$ from $F_{\delta^{(c_0)}}$, as values $p_i \in [0,1] \setminus [\delta^{(c_0)}, 1 - \delta^{(c_0)}]$ are disregarded for all $c$. 

A less obvious optimization would be to use a geometric progression of values $c$, e.g., $c \in \{2, 4, 8, 16, \ldots, c_0\}$ and maintain the user scores only for this set of coalition sizes, rather than for all values of $c \in \{2, \ldots, c_0\}$. This significantly reduces the space requirement per user from $O(c_0)$ to $O(\log c_0)$. However, if the actual coalition size is, say, $33$, then the coalition may not be caught until we reach $c = 64$. Since the codelength scales quadratically in $c$, this means that the codelength increases by a worst-case factor of $4$. In general, using any geometric progression with geometric factor $r$ possibly loses a factor $r^2$ in the codelengths. We have chosen to give the construction with many scores per user, to show that we then still obtain the same asymptotic codelengths. But the above construction is just one of the many alternatives to catch coalitions of any size efficiently. 

\subsection{Soundness}
\label{sub:UniversalTardos-Soundness}

For the universal Tardos scheme we get the following result regarding soundness.

\begin{theorem} \label{thm:PDS-UT-Soundness}
Consider the universal Tardos scheme in Section~\ref{sub:UniversalTardos-Construction}. If \eqref{eq:S'} is satisfied for each set of parameters $d_z^{(c)}, d_{\ell}^{(c)}, d_{\delta}^{(c)}, \eps_1^{(c)}$, and if the $\eps_1^{(c)}$ satisfy the following requirement:
\begin{align*}
\sum_{c=2}^{c_0} \eps_1^{(c)} \leq \eps_1, \tag{E} \label{eq:E}
\end{align*}
then the scheme is $\eps_1$-sound.
\end{theorem}

\begin{proof}
For each $c \in \{2, \ldots, c_0\}$, let $\hat{C}^{(c)}$ be the set of users that are accused because their scores $S_j^{(c)}$ exceeded $Z^{(c)}$ before $t^{(c)} > \ell^{(c)}$. Then $\hat{C} = \bigcup_{c=2}^{c_0} \hat{C}^{(c)}$. For any $c$, we can apply Theorem~\ref{thm:PDS-DT-Soundness} to the parameters $d_{\ell}^{(c)}, d_z^{(c)}, d_{\delta}^{(c)}, a^{(c)}$ and $\eps_1^{(c)}$ so that we know that the probability that $j \in \hat{C}^{(c)}$ for innocent users $j$ is at most $\eps_1^{(c)}/n$. So the overall probability that an innocent user is disconnected is bounded from above by
\begin{align*}
\pr(j \in \hat{C}) \leq \sum_{c=2}^{c_0} \pr(j \in \hat{C}^{(c)}) \leq \sum_{c=2}^{c_0} \frac{\eps_1^{(c)}}{n} \leq \frac{\eps_1}{n}.
\end{align*}
This completes the proof.
\end{proof}

Note that one can choose values $\eps_1^{(c)}$ satisfying \eqref{eq:E} such that $O(c^2 \ln(n/\eps_1^{(c)})) = O(c^2 \ln(n/\eps_1))$, e.g., by taking $\eps_1^{(c)} = 6\eps_1/(\pi^2 c^2)$. If furthermore $c = n^{o(1)}$ is subpolynomial in $n$, then asymptotically $d_{\ell} c^2 \ln(n/\eps_1^{(c)}) = d_{\ell} c^2 \ln(n/\eps_1) (1 + o(1))$ and we achieve the same asymptotic codelength as in the static and dynamic Tardos schemes.

\subsection{Dynamic completeness}
\label{sub:UniversalTardos-Completeness}

The main advantage of the universal Tardos scheme is that we can now prove dynamic completeness for all values of $c$.

\begin{theorem} \label{thm:PDS-UT-Completeness}
Consider the universal Tardos scheme in Section~\ref{sub:UniversalTardos-Construction}. If \eqref{eq:C'} is satisfied for each set of parameters $d_z^{(c)}, d_{\ell}^{(c)}, d_{\delta}^{(c)}, \eps_1^{(c)}, \eta^{(c)}$, where $\eta^{(c)} = \ln(1/\eps_2)/\ln(n/\eps_1^{(c)})$, then for each $c \in \{2, \ldots, c_0\}$ the scheme is dynamic $(\eps_2, c)$-complete.
\end{theorem}

\begin{proof}
This follows directly from applying Theorem~\ref{thm:PDS-DT-Completeness} to $d_{\ell}^{(c)}, d_z^{(c)}, d_{\delta}^{(c)}, a^{(c)}$ and $\eps_1^{(c)}$, where $c$ is the actual (unknown) coalition size.
\end{proof}

To prove that the scheme catches a coalition of size $c$, we only argued that the coalition's score $S^{(c)}(i)$ will exceed $c Z^{(c)}$ before we have seen $\ell^{(c)}$ positions $i$ with $p_i \in [\delta^{(c)}, 1 - \delta^{(c)}]$. In reality, the probability of catching the coalition is much larger than this, since for instance with high probability the coalition score $S^{(c-1)}$ will also exceed $Z^{(c-1)}$ before we have seen $\ell^{(c-1)}$ positions with $p_i \in [\delta^{(c-1)}, 1 - \delta^{(c-1)}]$. And if a pirate is disconnected because for some $k$ his score $S_j^{(k)}$ exceeded the threshold $Z^{(k)}$, then we do not have to wait until $S(i) > c\tilde{Z}^{(c)}$ but only until $S(i) > (c-1) Z^{(c)} + Z^{(k)}$. And since $S(i)$ has a constant slope, as soon as a pirate is caught, the other pirates' scores will increase even faster. In practice we therefore also see that we usually need fewer than $\ell^{(c)}$ positions to catch $c$ colluders.

\subsection{Codelengths}
\label{sub:UniversalTardos-Codelengths}

The theoretical results from the previous subsections are not for exactly $\ell^{(c)}$ watermark positions, but for some number of symbols $T^{(c)}$ such that there are $\ell^{(c)}$ positions $i$ between $1$ and $T^{(c)}$ with $p_i \in [\delta^{(c)}, 1 - \delta^{(c)}]$. The difference $T^{(c)} - \ell^{(c)}$ is a random variable, and is distributed according to a negative binomial distribution with parameters $r = \ell^{(c)}$ (the number of successes we are waiting for) and $p = 1 - \pr(p_i \in [\delta^{(c)}, 1 - \delta^{(c)}]) = \frac{4}{\pi} \arcsin(\sqrt{\delta^{(c)}})$ (the probability of a success). Because the parameter $p = O(c^{-2/3})$ is very small for large $c$, the difference between $T^{(c)}$ and $\ell^{(c)}$ will also be small. More precisely, $T^{(c)}$ has mean $\ell^{(c)}/(1 - p) = \ell^{(c)}(1 + O(c^{-2/3}))$ and variance $\sigma^2 = \ell^{(c)} p / (1 - p)^2 = O(\ell^{(c)} c^{-2/3})$, and the probability that $T^{(c)}$ exceeds its mean by $m > 0$ decreases exponentially in $m$. 

Also note that if some upper bound $c_0 \geq c$ is used for constructing the scheme as described earlier, and if the values of $p_i$ are drawn from $F_{\delta^{(c_0)}}$ instead of $F$, then we have $T^{(c_0)} = \ell^{(c_0)}$, as no values of $p_i$ are disregarded for $c = c_0$. So then the maximum codelength is fixed in advance, at the cost of possibly not catching coalitions of size $c > c_0$.

Finally, note that this scheme is constructed in such a way that coalitions of any (small) size can be caught more efficiently. To catch a coalition of size $c$ we now only use $O(c^2 \ln(n/\eps_1))$ symbols. This in comparison to the static and dynamic Tardos scheme, where we need $O(c_0^2 \ln(n/\eps_1))$ and $O(c c_0 \ln(n/\eps_1))$ symbols respectively, where $c_0$ is again some upper bound on the coalition size used to construct the schemes. So while using the dynamic Tardos scheme already reduces the codelength by a factor $\frac{c}{c_0}$, the universal Tardos scheme shaves off another factor $\frac{c}{c_0}$.

\subsection{Example}
\label{sub:UniversalTardos-Example}

As before, let the scheme parameters be given by $n = 10^6$ and $\eps_1 = \eps_2 = 10^{-3}$. Let us use $\eps_1^{(c)} = 6\eps_1/(\pi^2 c^2)$, so that $\sum_{c=2}^{c_0 = n} \eps_1^{(c)} \leq \eps_1$. Let us assume the coalition again has an actual size of $c = 25$. The optimal values of $d_{\ell}^{(25)}, d_z^{(25)}, d_{\delta}^{(25)}$ satisfying \eqref{eq:S'} and \eqref{eq:C'} can be calculated numerically as
\begin{align*}
  d_{\ell}^{(25)} = 8.59, \quad d_z^{(25)} = 4.61, \quad d_{\delta}^{(25)} = 13.83.
\end{align*}
This leads to the corresponding scheme parameters
\begin{align*}
  \ell^{(25)} = 148\,457, \ Z^{(25)} = 3188, \ \delta^{(25)} = 9.89 \cdot 10^{-4}.
\end{align*}
In Fig.~\ref{fig:Fig5} we show some simulation results for these parameters, where we only show the thresholds $Z^{(2)}, \ldots, Z^{(25)}$. In Fig.~\ref{fig:Fig5a} we simulated pirates using the interleaving attack, and in Fig.~\ref{fig:Fig5b} the pirates used the scapegoat strategy. As one can see, in the universal Tardos scheme the scapegoat strategy is not a good strategy, as the whole coalition is caught very soon. This is because the scapegoat strategy basically divides the coalition in $25$ coalitions of size $1$, and as mentioned before, small coalitions are caught much sooner in the universal Tardos scheme.

\begin{figure}[!t]
\centering
\subfloat[][Interleaving attack]{\includegraphics[width=\columnwidth]{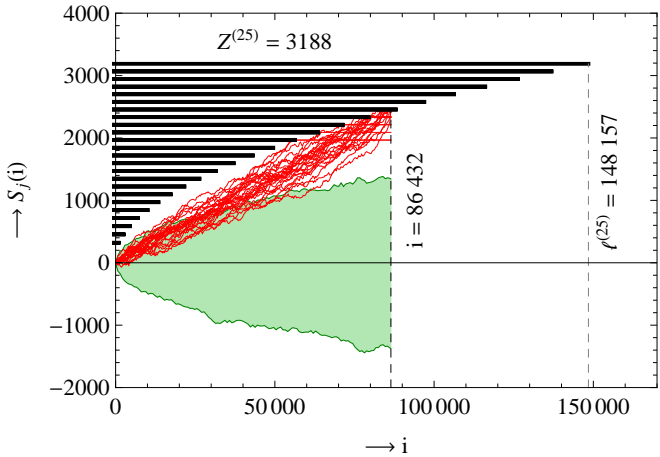} \label{fig:Fig5a}} \\
\subfloat[][Scapegoat strategy]{\includegraphics[width=\columnwidth]{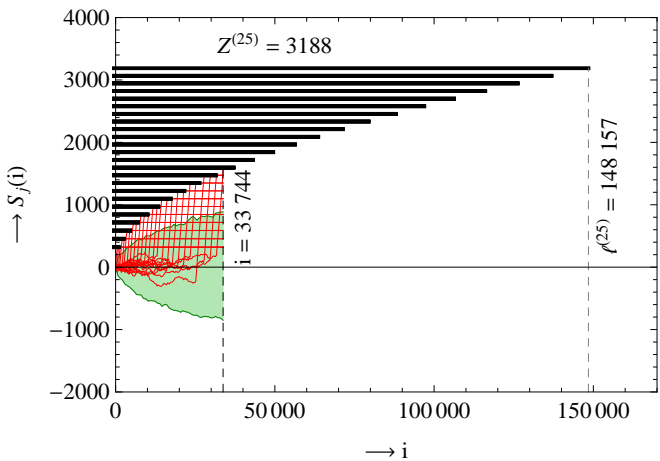} \label{fig:Fig5b}}
\caption{Simulations of the universal Tardos scheme, with parameters $c$, $c_0$, $n$, $\eps_1$, and $\eps_2$ as in Figs.~\ref{fig:Fig1} and \ref{fig:Fig3}. The black bars show the thresholds $Z^{(c)}$, for $c = 2, \ldots, 25$. For each pirate $j$ we only show the score $S_j^{(c)}(i)$ that made him get caught. In reality, all users have $25$ slightly different scores.}
\label{fig:Fig5}
\end{figure}

\section{Discussion}
\label{sec:Discussion}

Comparing the universal Tardos scheme to the static Tardos scheme, we see that the main advantages are that (a) we now have certainty about catching the whole coalition (instead of at least one pirate), and (b) we no longer need the coalition size, or a sharp upper bound on the coalition size, as input. We do need to calculate multiple scores per user, namely one for each possible coalition size $c$. But since the only disadvantage of a large $c_0$ is this larger number of scores per user and thus a larger offline space requirement (which may not be a big issue), $c_0$ can easily be much higher than the expected coalition size $c$. This in contrast to the static and dynamic Tardos schemes, where an increase in $c_0$ means an increase in the theoretical and practical codelengths as well.

In Table \ref{tab1} we list some of the differences between the static, dynamic, weakly dynamic and universal Tardos schemes. Here we assume that the upper bound $c_0$ on the number of colluders is the same for each scheme. The actual coalition size is denoted by $c$. The example referred to in the table is the example used throughout this paper, with $c = c_0 = 25$, $n = 10^6$, and $\eps_1 = \eps_2 = 10^{-3}$. The practical codelengths are based on $1000$ simulations for each scheme, where the pirates used the interleaving attack in all cases. For the weakly dynamic Tardos scheme we used $B = 8$ in our example.

\begin{table*}[!t]
\caption{A comparison of the Tardos schemes discussed in this paper. \label{tab1}}
\centering
\begin{tabular}{lccccc} \hline
 & static & dynamic & weakly dynamic & weakly dynamic & universal \\ 
 & (Section~\ref{sec:Preliminaries}) & (Section~\ref{sec:DynamicTardos}) & (Section~\ref{sub:SemiDynamicTardos-First}) & (Section~\ref{sub:SemiDynamicTardos-Second}) & (Section~\ref{sec:UniversalTardos}) \\ \hline 
scores per user & $1$ & $1$ & $1$ & $1$ & $c_0 - 1$ \\ 
density function & $f^{(c_0)}$ & $f^{(c_0)}$ & $f^{(c_0)}$ & $f^{(c_0)}$ & $f$ \\ 
blocks & $1$ of size $\ell$ & $\ell$ of size $1$ & $\ell/B$ of size $B$ & $\ell/B$ of size $B$ & $\ell$ of size $1$ \\ 
guilty caught & at least $1$ & all $c$ & all $c$ & all $c$ & all $c$ \\ 
expected codelength & $O(c_0^2 \ln(n/\eps_1))$ & $O(c c_0 \ln(n/\eps_1))$ & $O(c c_0 \ln(n/\eps_1))$ & $O(c c_0 \ln(n/\eps_1))$ & $O(c^2 \ln(n/\eps_1))$ \\ 
asymptotic codelength & $\frac{\pi^2}{2} c^2 \ln(n/\eps_1)$ & $\frac{\pi^2}{2} c^2 \ln(n/\eps_1)$ & $\frac{\pi^2}{2} c^2 \ln(n/\eps_1)$ & $\frac{\pi^2}{2} c^2 \ln(n/\eps_1)$ & $\frac{\pi^2}{2} c^2 \ln(n/\eps_1)$ \\ 
example, theoretical codelength & $109\,585$ & $116\,561$ & $116\,761$ & $131\,587$ & $148\,457$ \\ 
example, practical codelength & $109\,585$ & $92\,000$ & $92\,000$ & $96\,000$ & $89\,000$ \\ \hline
\end{tabular}
\end{table*}

Since our schemes are dynamic traitor tracing schemes, it makes sense to also compare them to other dynamic schemes from the literature. Recall from Section~\ref{sub:Introduction-RelatedWork} that the scheme of Fiat and Tassa~\cite{fiat01}, the schemes of Berkman et al.~\cite{berkman01} and the scheme of Roelse~\cite{roelse11} are deterministic schemes. That is, each of these schemes always catches all pirates and no user is ever falsely accused, which are advantages compared to probabilistic schemes such as our schemes. An additional advantage of these schemes is that they have very short codelengths. On the other hand, it was shown by Fiat and Tassa~\cite{fiat01} that $q \geq c+1$ for any deterministic scheme, so these schemes cannot be used in scenarios in which a small alphabet size is required. 

As is the case with our schemes, the dynamic scheme of Tassa~\cite{tassa05} is probabilistic and uses a binary alphabet (i.e., $q = 2$). The codelengths of these schemes can therefore be compared directly. In particular, the codelength of the scheme of Tassa is $\Theta(c^4 \log_2(n) \ln(n/\eps_1))$, which is more than a factor $\Theta(c^2)$ larger than the codelengths of our schemes. In fact, to the best of our knowledge our schemes have the shortest order codelengths of all known binary dynamic traitor tracing schemes.


Below we list some other nice properties of the universal Tardos scheme, which are not related to the codelength or the alphabet size. Most of these properties are inherited from the static Tardos scheme.


\paragraph{Codewords of users are independent.} This means that framing a specific innocent user is basically impossible, as the codewords of the pirates and the pirate output are independent of the innocent users' codewords. Also, a new user can be added to the system easily after the codewords of other users have already been generated, since the codewords of other users do not have to be updated.

\paragraph{Codeword positions are independent.} In other words, the scheme does not need the information obtained from the previous pirate output to generate new symbols for each user. Therefore the codewords can even be generated in advance. This also allows us to effectively tackle weakly dynamic traitor tracing scenarios, as described in Section~\ref{sec:SemiDynamicTardos}. 
In particular, the total tracing times of the dynamic schemes presented in \cite{berkman01, fiat01, roelse11, tassa05} are bounded from below by the total delay, defined as the codelength of the scheme times the delay of the pirates' transmission. By comparison, the total tracing times of our weakly dynamic schemes only increase marginally if $B$ increases. As a result, for a large delay (i.e. for a large value of $B$), our weakly dynamic schemes have the shortest total tracing times of all known dynamic schemes.
%

\paragraph{The distribution of watermark symbols is identical for each position.} This property offers new options, like tracing several coalitions simultaneously, using the same traitor tracing code. This also means that multiple watermarks from several broadcasts can be concatenated and viewed as one long watermark from one longer broadcast, allowing one to catch large coalitions with multiple watermarked broadcasts.

\paragraph{The codeword generation and accusation algorithm are computationally and memory-wise efficient.} The schemes do not require any complicated data structures and computations, and the only memory needed during the broadcast is the scores for each user at that time, and the counters $t^{(c)}$. During the broadcast only simple calculations are needed: computing $S_{j,i}$ (which has to be calculated only once), adding $S_{j,i}$ to those scores $S_j^{(c)}$ where $c$ satisfies a certain condition, and comparing the scores $S_j^{(c)}$ to the thresholds $Z^{(c)}$. 

\paragraph{Several instances of the scheme can be run simultaneously.} For example, by using parameters $\{\eps_1^{(c)}\}$ with $\sum \eps_1^{(c)} \leq 0.01$ and $\{\bar{\eps}_1^{(c)}\}$ with $\sum \bar{\eps}_1^{(c)} \leq 0.05$ for two different instances of the universal Tardos scheme (using the same codewords), a pirate will first cross one of the thresholds associated to $\{\bar{\eps}_1\}$, and only later cross one of the thresholds associated to the $\{\eps_1\}$. If we use the $\{\eps_1\}$ for disconnecting users, then even before a user is disconnected, we can give some sort of statistic to indicate the `suspiciousness' of this user. If a user then does not cross the highest thresholds, one could still decide whether to disconnect him or not. After all, the choice of $\eps_1$ may be arbitrary, and a user that almost crosses the thresholds $Z^{(c)}$ is likely to be guilty as well.

\section{Open problems}
\label{sec:OpenProblems}

Let us conclude with mentioning some open problems for future research.

\subsection{A single-score universal Tardos scheme}
\label{sub:OpenProblems-SingleScore}

Although we argued that the universal Tardos scheme has several advantages over other binary schemes, it has a minor drawback: we have to keep multiple scores for each user, namely for each possible coalition size $c$. To address this issue, one could try making small adjustments to the universal Tardos scheme, or start from the dynamic Tardos scheme and build a different, $c_0$-independent traitor tracing scheme. For instance, would it be possible to change the process of generating the $p_i$'s such that no positions are ever disregarded? Then all scores for one user would be the same, and we would only have to keep one score for each user.

\subsection{A continuous universal Tardos scheme} 
\label{sub:OpenProblems-Continuous}

Looking at Fig.~\ref{fig:Fig5} suggests that a continuous threshold function $Z(i)$ might also be an option, with $Z$ depending on the position $i$ instead of on the coalition size $c$. However, for the proof of soundness of the universal Tardos scheme, we simply added up the error probabilities for each threshold and showed that this sum is still less than $\eps_1$. If we use a continuous function $Z(i)$ and use this same proof method, this would lead to even smaller values of $\eps^{(i)}$ and longer codelengths. Still, theoretically it would be interesting to see if such a continuous threshold function can be constructed. 

\subsection{A fully dynamic Tardos scheme}
\label{sub:OpenProblems-FullyDynamic}

Most dynamic schemes find their strength in being able to adjust the next codeword symbols to the previous pirate output. In the dynamic Tardos scheme, we do not use this ability at all, and only use the dynamic setting to disconnect users inbetween. It is an open problem whether better results can be obtained with a fully dynamic Tardos scheme, that does use this extra power given to the distributor. 

\subsection{A weakly dynamic deterministic scheme}
\label{sub:OpenProblems-WeaklyDynamic}

The deterministic dynamic schemes in \cite{berkman01, fiat01, roelse11, tassa05} are not designed for the weakly dynamic setting, and it is not obvious how to adapt these schemes to this setting. The design and analysis of efficient weakly dynamic deterministic schemes is therefore an open problem.

\subsection{The dynamic traitor tracing capacity}
\label{sub:OpenProblems-Capacity}

On the other hand, it is also very well possible that no fully dynamic Tardos scheme exists that achieves significantly better codelengths. For the static setting, it is known that the order codelength of the Tardos scheme (quadratic in $c_0$, logarithmic in $n$) is optimal. But what about the dynamic setting? What is the optimal order codelength required to catch all colluders? Our results show that the optimal order codelength is at most quadratic in $c$, but this may not be optimal.

\subsection{A $q$-ary dynamic Tardos scheme} 
\label{sub:OpenProblems-qary}

In this paper we discussed several probabilistic dynamic schemes, taking the static binary Tardos scheme and the results of Laarhoven and De Weger~\cite{laarhoven11} as starting points. The design and analysis of $q$-ary probabilistic dynamic traitor tracing schemes is still an open problem. A possible approach for solving this problem is to take the $q$-ary Tardos scheme of \v{S}kori\'{c} et al.~\cite{skoric08} as a starting point. 


In a recent paper, Laarhoven et al.~\cite{laarhoven12} presented another approach to solve this problem. It was shown that with a divide-and-conquer construction, any binary dynamic traitor tracing scheme can be turned into a $q$-ary dynamic traitor tracing scheme with a codelength that is roughly a factor $q/2$ smaller than the codelength of the underlying binary scheme. Applying this to the constructions described in this paper, this leads to $q$-ary dynamic Tardos schemes with codelengths of the order $\ell_q = O\left(\frac{c^2}{q} \ln \frac{n}{\eps_1}\right)$. Moreover, for fixed $q$ and large $c$, this leads to an asymptotic codelength of $\ell_q \to \frac{\pi^2}{q} c^2 \ln \frac{n}{\eps_1}$, compared to the $\ell_2 \to \frac{\pi^2}{2} c^2 \ln \frac{n}{\eps_1}$ of the binary schemes presented in this paper. For details, see \cite{laarhoven12}. 


\section*{Acknowledgments} 
The authors are grateful to the anonymous reviewers for their valuable comments.

\end{document}